\documentclass[letterpaper, 10 pt, dvipsnames, conference]{ieeeconf}
\usepackage{graphicx,url}
\graphicspath{ {./images/} }
\usepackage{amssymb,amsmath,mathtools}
\usepackage{tikz,pgfplots}
\usepackage{pgfplots}
\usepackage{pgfplotstable}
\usetikzlibrary{pgfplots.groupplots}
\usepgfplotslibrary{colorbrewer}
\pgfplotsset{compat = 1.15, cycle list/Set1-8}
\usetikzlibrary{pgfplots.statistics, pgfplots.colorbrewer}
\usetikzlibrary{pgfplots.groupplots}
\usetikzlibrary{shapes.geometric,backgrounds,patterns, trees}
\usetikzlibrary{3d,decorations.text,shapes.arrows,positioning,fit,backgrounds}
\usetikzlibrary{positioning, decorations.pathmorphing, shapes}
\usetikzlibrary{decorations.pathreplacing}
\usetikzlibrary{shapes.geometric,backgrounds,patterns, trees}
\usetikzlibrary{spy}
\usetikzlibrary{arrows.meta,
                bending,
                intersections,
                quotes,
                shapes.geometric}
              \usetikzlibrary{automata, positioning}
              \usepgfplotslibrary{fillbetween}
\usetikzlibrary{shapes,arrows}
\usetikzlibrary{arrows.meta}
\usetikzlibrary{positioning}
\tikzset{set/.style={draw,circle,inner sep=0pt,align=center}}
\usetikzlibrary{automata, positioning}
  \usetikzlibrary{shapes,shadows}
  \tikzstyle{abstractbox} = [draw=black, fill=white, rectangle,
  inner sep=10pt, style=rounded corners, drop shadow={fill=black,
  opacity=1}]
\tikzstyle{abstracttitle} =[fill=white]
\usetikzlibrary{calc,positioning,shapes.geometric}
\usetikzlibrary{arrows.meta,arrows}
\usetikzlibrary{matrix}

\colorlet{myRed}{red!20}
\tikzset{
  rows/.style 2 args={/utils/temp/.style={row ##1/.append style={nodes={#2}}},
    /utils/temp/.list={#1}},
  columns/.style 2 args={/utils/temp/.style={column ##1/.append style={nodes={#2}}},
    /utils/temp/.list={#1}}}
\usetikzlibrary{backgrounds,calc,shadings,shapes.arrows,shapes.symbols,shadows}
\definecolor{switch}{HTML}{006996}
\pgfkeys{/pgf/.cd,
  parallelepiped offset x/.initial=2mm,
  parallelepiped offset y/.initial=2mm
}
\pgfdeclareshape{parallelepiped}
{
  \inheritsavedanchors[from=rectangle] 
  \inheritanchorborder[from=rectangle]
  \inheritanchor[from=rectangle]{north}
  \inheritanchor[from=rectangle]{north west}
  \inheritanchor[from=rectangle]{north east}
  \inheritanchor[from=rectangle]{center}
  \inheritanchor[from=rectangle]{west}
  \inheritanchor[from=rectangle]{east}
  \inheritanchor[from=rectangle]{mid}
  \inheritanchor[from=rectangle]{mid west}
  \inheritanchor[from=rectangle]{mid east}
  \inheritanchor[from=rectangle]{base}
  \inheritanchor[from=rectangle]{base west}
  \inheritanchor[from=rectangle]{base east}
  \inheritanchor[from=rectangle]{south}
  \inheritanchor[from=rectangle]{south west}
  \inheritanchor[from=rectangle]{south east}
  \backgroundpath{
    \southwest \pgf@xa=\pgf@x \pgf@ya=\pgf@y
    \northeast \pgf@xb=\pgf@x \pgf@yb=\pgf@y
    \pgfmathsetlength\pgfutil@tempdima{\pgfkeysvalueof{/pgf/parallelepiped
      offset x}}
    \pgfmathsetlength\pgfutil@tempdimb{\pgfkeysvalueof{/pgf/parallelepiped
      offset y}}
    \def\ppd@offset{\pgfpoint{\pgfutil@tempdima}{\pgfutil@tempdimb}}
    \pgfpathmoveto{\pgfqpoint{\pgf@xa}{\pgf@ya}}
    \pgfpathlineto{\pgfqpoint{\pgf@xb}{\pgf@ya}}
    \pgfpathlineto{\pgfqpoint{\pgf@xb}{\pgf@yb}}
    \pgfpathlineto{\pgfqpoint{\pgf@xa}{\pgf@yb}}
    \pgfpathclose
    \pgfpathmoveto{\pgfqpoint{\pgf@xb}{\pgf@ya}}
    \pgfpathlineto{\pgfpointadd{\pgfpoint{\pgf@xb}{\pgf@ya}}{\ppd@offset}}
    \pgfpathlineto{\pgfpointadd{\pgfpoint{\pgf@xb}{\pgf@yb}}{\ppd@offset}}
    \pgfpathlineto{\pgfpointadd{\pgfpoint{\pgf@xa}{\pgf@yb}}{\ppd@offset}}
    \pgfpathlineto{\pgfqpoint{\pgf@xa}{\pgf@yb}}
    \pgfpathmoveto{\pgfqpoint{\pgf@xb}{\pgf@yb}}
    \pgfpathlineto{\pgfpointadd{\pgfpoint{\pgf@xb}{\pgf@yb}}{\ppd@offset}}
  }
}

\makeatletter
\tikzset{anchor/.append code=\let\tikz@auto@anchor\relax,
  add font/.code=%
    \expandafter\def\expandafter\tikz@textfont\expandafter{\tikz@textfont#1},
  left delimiter/.style 2 args={append after command={\tikz@delimiter{south east}
    {south west}{every delimiter,every left delimiter,#2}{south}{north}{#1}{.}{\pgf@y}}}}
\tikzstyle{sms} = [rectangle callout, draw,very thick, rounded corners, minimum height=20pt]
\makeatletter
\tikzset{anchor/.append code=\let\tikz@auto@anchor\relax,
  add font/.code=%
    \expandafter\def\expandafter\tikz@textfont\expandafter{\tikz@textfont#1},
  left delimiter/.style 2 args={append after command={\tikz@delimiter{south east}
    {south west}{every delimiter,every left delimiter,#2}{south}{north}{#1}{.}{\pgf@y}}}}
\tikzstyle{sms} = [rectangle callout, draw,very thick, rounded corners, minimum height=20pt]
\usetikzlibrary{positioning,calc}
\tikzstyle{block} = [rectangle, draw,
text width=10.5em, text centered, rounded corners, minimum height=4em]
\tikzstyle{line} = [draw, -latex]
\tikzset{l3 switch/.style={
    parallelepiped,fill=switch, draw=white,
    minimum width=0.75cm,
    minimum height=0.75cm,
    parallelepiped offset x=1.75mm,
    parallelepiped offset y=1.25mm,
    path picture={
      \node[fill=white,
        circle,
        minimum size=6pt,
        inner sep=0pt,
        append after command={
          \pgfextra{
            \foreach \angle in {0,45,...,360}
            \draw[-latex,fill=white] (\tikzlastnode.\angle)--++(\angle:2.25mm);
          }
        }
      ]
       at ([xshift=-0.75mm,yshift=-0.5mm]path picture bounding box.center){};
    }
  },
  ports/.style={
    line width=0.3pt,
    top color=gray!20,
    bottom color=gray!80
  },
  rack switch/.style={
    parallelepiped,fill=white, draw,
    minimum width=1.25cm,
    minimum height=0.25cm,
    parallelepiped offset x=2mm,
    parallelepiped offset y=1.25mm,
    xscale=-1,
    path picture={
      \draw[top color=gray!5,bottom color=gray!40]
      (path picture bounding box.south west) rectangle
      (path picture bounding box.north east);
      \coordinate (A-west) at ([xshift=-0.2cm]path picture bounding box.west);
      \coordinate (A-center) at ($(path picture bounding box.center)!0!(path
        picture bounding box.south)$);
      \foreach \x in {0.275,0.525,0.775}{
        \draw[ports]([yshift=-0.05cm]$(A-west)!\x!(A-center)$)
          rectangle +(0.1,0.05);
        \draw[ports]([yshift=-0.125cm]$(A-west)!\x!(A-center)$)
          rectangle +(0.1,0.05);
       }
      \coordinate (A-east) at (path picture bounding box.east);
      \foreach \x in {0.085,0.21,0.335,0.455,0.635,0.755,0.875,1}{
        \draw[ports]([yshift=-0.1125cm]$(A-east)!\x!(A-center)$)
          rectangle +(0.05,0.1);
      }
    }
  },
  server/.style={
    parallelepiped,
    fill=white, draw,
    minimum width=0.35cm,
    minimum height=0.75cm,
    parallelepiped offset x=3mm,
    parallelepiped offset y=2mm,
    xscale=-1,
    path picture={
      \draw[top color=gray!5,bottom color=gray!40]
      (path picture bounding box.south west) rectangle
      (path picture bounding box.north east);
      \coordinate (A-center) at ($(path picture bounding box.center)!0!(path
        picture bounding box.south)$);
      \coordinate (A-west) at ([xshift=-0.575cm]path picture bounding box.west);
      \draw[ports]([yshift=0.1cm]$(A-west)!0!(A-center)$)
        rectangle +(0.2,0.065);
      \draw[ports]([yshift=0.01cm]$(A-west)!0.085!(A-center)$)
        rectangle +(0.15,0.05);
      \fill[black]([yshift=-0.35cm]$(A-west)!-0.1!(A-center)$)
        rectangle +(0.235,0.0175);
      \fill[black]([yshift=-0.385cm]$(A-west)!-0.1!(A-center)$)
        rectangle +(0.235,0.0175);
      \fill[black]([yshift=-0.42cm]$(A-west)!-0.1!(A-center)$)
        rectangle +(0.235,0.0175);
    }
  },
}

\usetikzlibrary{calc, shadings, shadows, shapes.arrows}
\tikzset{cross/.style={cross out, draw=black, minimum size=2*(#1-\pgflinewidth), inner sep=0pt, outer sep=0pt},
cross/.default={1pt}}
\tikzset{%
  interface/.style={draw, rectangle, rounded corners, font=\LARGE\sffamily},
  ethernet/.style={interface, fill=yellow!50},
  serial/.style={interface, fill=green!70},
  speed/.style={sloped, anchor=south, font=\large\sffamily},
  route/.style={draw, shape=single arrow, single arrow head extend=4mm,
    minimum height=1.7cm, minimum width=3mm, white, fill=switch!20,
    drop shadow={opacity=.8, fill=switch}, font=\tiny}
}

\usepackage{float}
\definecolor{bluetwo}{RGB}{189, 213, 234}
\definecolor{bluethree}{RGB}{165, 193, 224}
\definecolor{bluefour}{RGB}{141, 169, 200}

\usepackage{amsthm}
\newtheoremstyle{ieeebold}
  {5pt}                   
  {5pt}                   
  {\normalfont}           
  {}                      
  {\bfseries}             
  {.}                     
  {0.5em}                 
  {}                    
\newtheorem{observation}{Observation}
\theoremstyle{ieeebold}
\newtheorem{proposition}{Proposition}

\newtheorem{lemma}{Lemma}

\newtheorem{remark}{Remark}

\usepackage{bbm}
\usepackage{bm}

\newcommand\numeq[1]%
{\stackrel{\scriptscriptstyle(\mkern-1.5mu#1\mkern-1.5mu)}{=}}
\newcommand\numeqq[1]%
{\stackrel{\scriptscriptstyle(\mkern-1.5mu#1\mkern-1.5mu)}{\triangleq}}
\newcommand\numleq[1]%
{\stackrel{\scriptscriptstyle(\mkern-1.5mu#1\mkern-1.5mu)}{\leq}}
\newcommand\numgeq[1]%
{\stackrel{\scriptscriptstyle(\mkern-1.5mu#1\mkern-1.5mu)}{\geq}}
\newcommand\numimp[1]%
{\stackrel{\scriptscriptstyle(\mkern-1.5mu#1\mkern-1.5mu)}{\implies}}

\usetikzlibrary{backgrounds}
\usetikzlibrary{patterns}
\pgfmathdeclarefunction{gauss}{3}{%
  \pgfmathparse{1/(#3*sqrt(2*pi))*exp(-((#1-#2)^2)/(2*#3^2))}%
}

\usepackage{listofitems} 
\usetikzlibrary{arrows.meta} 
\usepackage[outline]{contour} 
\contourlength{1.4pt}

\colorlet{myred}{red!80!black}
\colorlet{myblue}{blue!80!black}
\colorlet{mygreen}{green!60!black}
\colorlet{myorange}{orange!70!red!60!black}
\colorlet{mydarkred}{red!30!black}
\colorlet{mydarkblue}{blue!40!black}
\colorlet{mydarkgreen}{green!30!black}

\tikzset{
  >=latex, 
  node/.style={thick,circle,draw=myblue,minimum size=22,inner sep=0.5,outer sep=0.6},
  node in/.style={node,black!20!black,draw=mygreen!30!black,fill=black!20},
  node hidden/.style={node,black!20!black,draw=myblue!30!black,fill=black!20},
  node convol/.style={node,black!20!black,draw=myorange!30!black,fill=black!20},
  node out/.style={node,red!20!black,draw=myred!30!black,fill=black!20},
  connect/.style={thick,Blue!100}, 
  connect arrow/.style={-{Latex[length=4,width=3.5]},thick,mydarkblue,shorten <=0.5,shorten >=1},
  node 1/.style={node in}, 
  node 2/.style={node hidden},
  node 3/.style={node out}
}

\DeclareMathOperator*{\argmax}{arg\,max}

\tikzset{
    database/.style={
        path picture={
            \draw (0, 1.5*\database@segmentheight) circle [x radius=\database@radius,y radius=\database@aspectratio*\database@radius];
            \draw (-\database@radius, 0.5*\database@segmentheight) arc [start angle=180,end angle=360,x radius=\database@radius, y radius=\database@aspectratio*\database@radius];
            \draw (-\database@radius,-0.5*\database@segmentheight) arc [start angle=180,end angle=360,x radius=\database@radius, y radius=\database@aspectratio*\database@radius];
            \draw (-\database@radius,1.5*\database@segmentheight) -- ++(0,-3*\database@segmentheight) arc [start angle=180,end angle=360,x radius=\database@radius, y radius=\database@aspectratio*\database@radius] -- ++(0,3*\database@segmentheight);
        },
        minimum width=2*\database@radius + \pgflinewidth,
        minimum height=3*\database@segmentheight + 2*\database@aspectratio*\database@radius + \pgflinewidth,
    },
    database segment height/.store in=\database@segmentheight,
    database radius/.store in=\database@radius,
    database aspect ratio/.store in=\database@aspectratio,
    database segment height=0.1cm,
    database radius=0.25cm,
    database aspect ratio=0.35,
  }

\definecolor{gray2}{HTML}{ededed}
\definecolor{gray3}{HTML}{F5F5F5}
\definecolor{RoyalAzure}{rgb}{0.0, 0.22, 0.66}
\definecolor{lightgray}{gray}{0.9}
\definecolor{lightgray}{gray}{0.9}
\definecolor{lightgreen}{rgb}{0.88, 1, 0.88}
\definecolor{lightred}{rgb}{1, 0.88, 0.88}
\definecolor{lightblue}{rgb}{0.88, 0.94, 1}
\definecolor{lightorange}{rgb}{1, 0.94, 0.88}

\usetikzlibrary{calc, shadings, shadows, shapes.arrows}
\tikzset{cross/.style={cross out, draw=black, minimum size=2*(#1-\pgflinewidth), inner sep=0pt, outer sep=0pt},
cross/.default={1pt}}
\tikzset{%
  interface/.style={draw, rectangle, rounded corners, font=\LARGE\sffamily},
  ethernet/.style={interface, fill=yellow!50},
  serial/.style={interface, fill=green!70},
  speed/.style={sloped, anchor=south, font=\large\sffamily},
  route/.style={draw, shape=single arrow, single arrow head extend=4mm,
    minimum height=1.7cm, minimum width=3mm, white, fill=switch!20,
    drop shadow={opacity=.8, fill=switch}, font=\tiny}
}
\usepackage{dsfont}
\allowdisplaybreaks

\usepackage{mathrsfs}
\usepackage{colortbl}
\usepackage{booktabs}

\IEEEoverridecommandlockouts
\title{\LARGE \bf
Optimal Stopping of Self-Refining Foundation Models
}

\author{Kim Hammar, Tansu Alpcan, and Emil C. Lupu
\thanks{K. Hammar is supported by the Swedish Research Council under contract 2024-06436.}
\thanks{
T. Alpcan is with the University of Melbourne, Australia. \texttt{tansu.alpcan@unimelb.edu.au}.
}
\thanks{
K. Hammar and E.C. Lupu are with Imperial College London, United Kingdom. \texttt{\{k.hammar,e.c.lupu\}@imperial.ac.uk}.
}
}

\begin{document}

\maketitle
\thispagestyle{empty}
\pagestyle{empty}

\begin{abstract}
Foundation models can improve their outputs through a self-refinement process driven by external feedback. In this process, the model is embedded in an iterative loop where it generates outputs, receives feedback from verifiers, and refines its responses through in-context learning. Following a novel approach, we formalize this process as an optimal stopping problem where the number of refinement iterations is decided based on expected improvement relative to cost. We derive optimal stopping policies and show that they can be efficiently computed through stochastic approximation. To evaluate our approach experimentally, we apply it to a coding benchmark for foundation models. The empirical results show that our stopping policies are  significantly more cost-efficient than stopping policies proposed in prior work.
\end{abstract}

\section{Introduction}
Foundation models are becoming an important component of decision-making systems across many domains, including software development \cite{he2025llm}, scientific research \cite{Lu2026}, content generation \cite{maleki2024procedural}, and systems engineering \cite{hammar2025incidentresponseplanningusing}. They consist of large neural networks that are trained on vast datasets (e.g., web-scale text corpora), which enables them to generalize across tasks and modalities. Prominent examples include large language models (e.g., Gemini \cite{geminiteam2024geminifamilyhighlycapable}), time-series models (e.g., Chronos \cite{pmlr-v258-arango25a}), and visual models (e.g., Flamingo \cite{NEURIPS2022_960a172b}).

These types of models generate outputs autoregressively by sampling from a conditional distribution $p_{\theta}(z \mid v)$, where $\theta$ denotes the model parameters, $v$ is an input (e.g., a task description), and $z$ is the generated output (e.g., code or text). Since the model adapts to the input $v$, the same model can be applied to different tasks without updating the parameters $\theta$, a property known as \textit{in-context learning} \cite{NEURIPS2020_1457c0d6}. This flexibility also enables the model to refine its outputs based on feedback. In particular, given an initial output $z$ and feedback $x$ (e.g., from automated verification procedures), the model can generate a revised output by sampling from $p_{\theta}(z' \mid v,z,x)$, as initially shown by Madaan et al. \cite{madaan2023selfrefine}.

The process of generating and revising outputs based on feedback typically continues for a fixed number of iterations or until a predetermined stopping criterion is met. While prior work has demonstrated that this approach improves outputs across various application domains, such as coding (see e.g., \cite{NEURIPS2024_99c66755}) and logical reasoning (see e.g., \cite{weng-etal-2023-large,wang2024a}), the literature is focused on empirical evaluation and lacks a formal analysis. In particular, there is little understanding of how to optimally decide when to stop refining.

In this paper, we address this research gap by presenting a decision-theoretic model of the self-refinement process. Specifically, we formulate self-refinement as an optimal stopping problem in which the number of refinement iterations is decided based on the expected improvement relative to the cost of invoking the foundation model. For example, if the foundation model is hosted locally, each invocation incurs computational expenditure. Similarly, if the model is accessed via an external provider, it incurs a monetary cost.

\begin{figure}
  \centering
  \scalebox{0.72}{
    \includegraphics{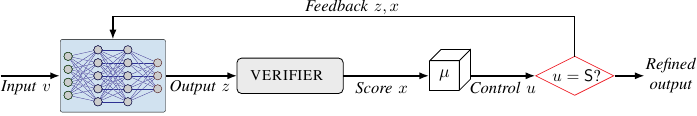}
  }
\caption{Illustration of a self-refining foundation model with stopping. The model is integrated into an iterative verification and refinement loop, in which it is used to generate outputs that are verified against constraints to generate a score. The score is then input to a stopping policy $\mu$ that decides whether to stop and accept the output or to continue revising it.}
  \label{fig:method}
\end{figure}

Leveraging this mathematical formulation, we establish conditions for optimal threshold-based stopping policies and validate them on a coding benchmark across three frontier models. Empirical evaluations show that our stopping policies are cost-efficient and consistently improve performance when compared to other approaches proposed in the research literature, which rely on heuristic stopping policies. 

Our contributions can be summarized as follows:
\begin{itemize}
\item We present a novel formulation of self-refinement in foundation models as an optimal stopping problem.
\item We derive optimal stopping policies and validate them on a coding benchmark. The results show that they are more cost-efficient than stopping policies proposed in prior work. Our implementation is available at \cite{cdc_source_kim}.
\end{itemize}  

\section{Related Work}
Optimal stopping is a classical problem with a well-developed theory; see e.g., Wald \cite{wald}, Shiryaev \cite{shirayev}, and Chow et al. \cite{chow1971great}. Example use cases for this theory include asset selling \cite{bert05}, intrusion detection \cite{tartakovsky_1}, network management \cite{hammar_stadler_tnsm}, machine replacement \cite{krishnamurthy_2016}, hypothesis testing \cite{wald}, gambling \cite{chow1971great}, industrial control \cite{kalle_stopping_industrial}, change detection \cite{10423411,7569076,9086257,8550188,10.1109/CDC42340.2020.9303818}, and flow control \cite{10955193,11592598,Hammar_Alpcan_2026}. To our knowledge, we are the first to apply this theory to analyze the self-refinement process of a foundation model.

Many variants of the stopping problem have been studied. Examples include discrete-time and continuous-time problems, finite-horizon and infinite-horizon problems, problems with fully observed and partially observed state spaces, problems with finite and infinite state spaces, Markovian and non-Markovian problems, and single-stop and multi-stop problems. Consequently, different solution approaches have been developed. The most common are the \textit{martingale} approach (see e.g., Snell \cite{Snell1952TAMS}) and the \textit{Markovian} approach (see e.g., Bather \cite{bather_decision_theory}). In this paper, we investigate a stopping problem with a finite time horizon, discrete-time progression, a continuous state space, and the Markov property.

\section{The Self-Refinement Use Case}
We consider a scenario where a foundation model is used to solve a task specified by an instruction in natural language. As an example, the task could be to solve a programming problem and the model could be a large language model, such as Gemini \cite{geminiteam2024geminifamilyhighlycapable}. To evaluate the solution produced by the model, we associate it with a \textit{score} $x \in [0,1]$, with $x=1$ being the optimal score. For instance, in the programming example, the solution could be a piece of code and the score $x$ could capture its correctness and computational efficiency.

We assume that the model is embedded in an iterative loop where it refines its output based on feedback, as illustrated in Fig.~\ref{fig:method}. Each iteration invokes the model with the current score $x \in [0,1]$ as feedback, which leads to a new output that receives an updated score $x'$ and incurs a cost $c > 0$. This cost can be either computational (if the model is run locally) or monetary (if the model is provided externally). 

After each refinement iteration, a stopping policy $\mu$ uses the current score to decide whether to stop and accept the output or to continue refining, where the maximum number of iterations is $N$. When designing this policy, the goal is to optimally balance the refinement costs against the potential improvement in the score $x$, as formally defined below.
\section{Optimal Stopping Formulation}
We formalize the self-refinement use case described above as a discrete-time dynamical system where the score evolves as a Markov process $(x_{k})_{k=0}^{N}$. At each stage $k \in \{0,1,\hdots,N-1\}$, two controls are available: ($\mathsf{S}$)top and ($\mathsf{C}$)ontinue. Control $u=\mathsf{S}$ in state $x$ yields a payoff $g(x) \geq 0$ that quantifies the quality of the output and terminates the process. Conversely, control $u=\mathsf{C}$ incurs a cost $c > 0$ and transitions the system to the next state according to 
\begin{equation}\label{eq:dynamics}
x_{k+1} = f(x_k, w_k), \quad k = 0, 1, \ldots, N-1,
\end{equation}
where the initial state $x_0$ is given, $(w_k)_{k=0}^{N-1}$ are i.i.d. random variables with distribution $P_w$ that capture the stochasticity of the generation process, and $f$ is the system function.

Let $(\Omega, \mathcal{F}, P)$ be the underlying probability space and let $(\mathcal{F}_k)^{N}_{k=0}$ denote the filtration generated by the quality scores, i.e., each $\mathcal{F}_k$ is the sigma-algebra generated by the random variables $x_0, \hdots, x_k$. Within this space, we define a stopping policy $\mu = (\mu_0, \mu_1, \hdots, \mu_{N})$ as a sequence of measurable functions $\mu_k: [0,1] \rightarrow \{\mathsf{S},\mathsf{C}\}$ that map the state to a control.

Given a policy $\mu$, we define the stopping time as
\begin{equation}\label{eq:tau}
\tau_{\mu} = \min\{k \in \{0, 1, \ldots, N\} \mid  \mu_k(x_k) = \mathsf{S}\},
\end{equation}
with the constraint $\mu_N(x_N) = \mathsf{S}$, which ensures $\tau_{\mu} \leq N$. 

\begin{proposition}\label{prop:stopping_time}
The random variable $\tau_{\mu}$ defined in \eqref{eq:tau} is a stopping time with respect to the filtration $(\mathcal{F}_k)^N_{k=0}$.
\end{proposition}
\begin{proof}
It suffices to show that the event $\{\tau_{\mu} = k\}$ belongs to the sigma-algebra $\mathcal{F}_k$ for all $k \in \{0,1,\hdots,N\}$. For $k < N$, $\{\tau_{\mu}=k\}$ is the event that the stopping policy continues at all preceding iterations and stops at iteration $k$, i.e.,
\begin{align*}
\{\tau_{\mu}=k\} = \Bigg(\bigcap_{j=0}^{k-1} \{\mu_j(x_j) = \mathsf{C}\}\Bigg) \cap \{\mu_k(x_k) = \mathsf{S}\},
\end{align*}
where the intersection is $\Omega$ when $k=0$.

Since each policy $\mu_j$ is a measurable function of $x_j$ and $x_j$ is $\mathcal{F}_j$-measurable (and hence $\mathcal{F}_k$-measurable for $j \leq k$), each set in the intersection belongs to $\mathcal{F}_k$. As a consequence, we have $\{\tau_{\mu}=k\} \in \mathcal{F}_k$. For $k=N$, the event $\{\tau_{\mu}=N\}=\bigcap_{j=0}^{N-1}\{\mu_j(x_j) = \mathsf{C}\}$ belongs to $\mathcal{F}_{N-1} \subseteq \mathcal{F}_N$. Thus $\tau_{\mu}$ is a stopping time with respect to $(\mathcal{F}_k)^{N}_{k=0}$.
\end{proof}
Given the preceding definition of a stopping policy $\mu$, we define the stopping and continuation sets as
\begin{align*}
\mathscr{S}_k^{\mu} &= \{x \mid x \in [0,1], \mu_k(x) = \mathsf{S}\}, && (\text{Stopping set})\\
\mathscr{C}_k^{\mu} &= \{x \mid x \in [0,1], \mu_k(x) = \mathsf{C}\}. && (\text{Continuation set})
\end{align*}  

When designing the stopping policy $\mu$, the objective is to maximize the expected payoff $E\{g(x_{\tau_{\mu}})\}$ while minimizing refinement costs. In particular, each refinement iteration incurs a cost $c > 0$ (e.g., per-token charges) and the objective is to find an optimal stopping policy $\mu^{\star}$ that satisfies
\begin{equation}\label{eq:objective}
\mu^\star \in \argmax_{\mu}  E\left\{g(x_{\tau_{\mu}}) -c\tau_{\mu} \right\} \text{ subject to \eqref{eq:dynamics}},
\end{equation}
where the expectation is taken over the randomness in the dynamics \eqref{eq:dynamics} and the stopping time $\tau_{\mu}$ is defined in \eqref{eq:tau}. We denote the optimal stopping time, stopping set, and continuation set by $\tau^{\star}, \mathscr{S}_k^{\star}$, and $\mathscr{C}_k^{\star}$, respectively.
\begin{remark}
\textit{We model the cost $c$ as a constant. In principle, $c$ could vary with the stage $k$, e.g., if token consumption changes across iterations. However, in our experimental evaluation (see \S\ref{sec:id}), token consumption per iteration is approximately constant, which justifies this assumption.}
\end{remark}
\subsection*{Characterizing the Optimal Stopping Policy}
To characterize the optimal stopping policy for problem \eqref{eq:objective}, we use a dynamic programming formulation. Following this approach, we define the stage-dependent stopping time
\begin{align*}
\tau_{\mu,k} = \min\{j \in \{k, k+1, \hdots, N\} \mid \mu_j(x_j) = \mathsf{S}\},
\end{align*}
so that $\tau_{\mu} = \tau_{\mu,0}$; cf.~\eqref{eq:tau}. We then define the value function $V^{\mu}_k: [0,1] \rightarrow \Re$ to represent the expected payoff minus refinement costs from state $x$ at stage $k$ under policy $\mu$, i.e.,
\begin{align*}
V^{\mu}_k(x) &= E\{g(x_{\tau_{\mu,k}})-c(\tau_{\mu,k}-k) \mid x_k=x\},
\end{align*}
for all stages $k=0, 1, \ldots, N-1$ and states $x \in [0,1]$, with the terminal condition $V^{\mu}_N(x)=g(x)$ for all $x \in [0,1]$.

The optimal value function $V_k^{\star}$, which is derived by optimizing over all policies $\mu$, satisfies the Bellman equation
\begin{equation}\label{eq:bellman}
V^{\star}_k(x) = \max\Big\{g(x),\; -c + E\big\{V^{\star}_{k+1}(f(x, w_k))\big\}\Big\},
\end{equation}
for all stages $k=0, 1, \ldots, N-1$ and states $x \in [0,1]$.

The first term in the maximization \eqref{eq:bellman} corresponds to the decision to stop ($u=\mathsf{S}$), which yields the payoff $g(x)$. The second term corresponds to the decision to continue ($u=\mathsf{C}$), which incurs the cost $c$ and leads to a new state. We denote the value of continuing under an optimal policy by
\begin{align*}
Q^{\star}_k(x) = -c + E\{V^{\star}_{k+1}(f(x, w_k))\}.
\end{align*}
Given this notation, the optimal policy can be expressed as
\begin{align*}
  \mu_k^{\star}(x) &=
                     \begin{dcases}
                       \mathsf{S}, & \text{if } g(x) \geq Q^{\star}_k(x),\\
                       \mathsf{C}, & \text{otherwise},
                     \end{dcases}  && k=0,1,\hdots,N-1.
\end{align*}
The above structure implies that the sequence $(V^{\star}_k)_{k=0}^{N}$ is the smallest superharmonic majorant of $g$.\footnote{A sequence $(h_k)_{k=0}^{N}$ is \emph{superharmonic} with respect to the system $f(x,w)$ if $h_k(x) \geq -c + E\{h_{k+1}(f(x,w))\}$ for all $x$ and $k < N$. It is a \emph{superharmonic majorant} of $g$ if additionally $h_k(x) \geq g(x)$ everywhere.} In other words, the stopping set is the set of states where $V^{\star}_k$ equals $g$, while the continuation set is the set where $V^{\star}_k$ exceeds $g$. This characterization admits a geometric interpretation of the optimal stopping policy, as illustrated in Fig.~\ref{fig:harmonic}.

\begin{figure}[H]
  \centering
  \scalebox{0.75}{
    \includegraphics{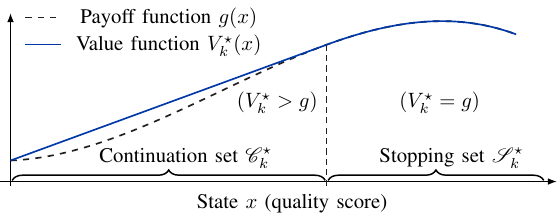}
  }
  \caption{Geometric characterization of the optimal stopping policy.}
  \label{fig:harmonic}
\end{figure}
\begin{remark}
\textit{When $c = 0$, the Bellman equation \eqref{eq:bellman} reduces to $V^{\star}_k(x) = \max\{g(x), E\{V^{\star}_{k+1}(f(x, w_k))\}\}$. In this case, the process $(V^{\star}_k(x_k))_{k=0}^{N}$ is the Snell envelope of the payoff process $(g(x_k))_{k=0}^{N}$ \cite{Snell1952TAMS} and the problem corresponds to the classical optimal stopping formulation of Dynkin~\cite{dynkin1963optimum}.}
\end{remark}
\begin{figure*}
  \centering
  \scalebox{0.74}{
    \includegraphics{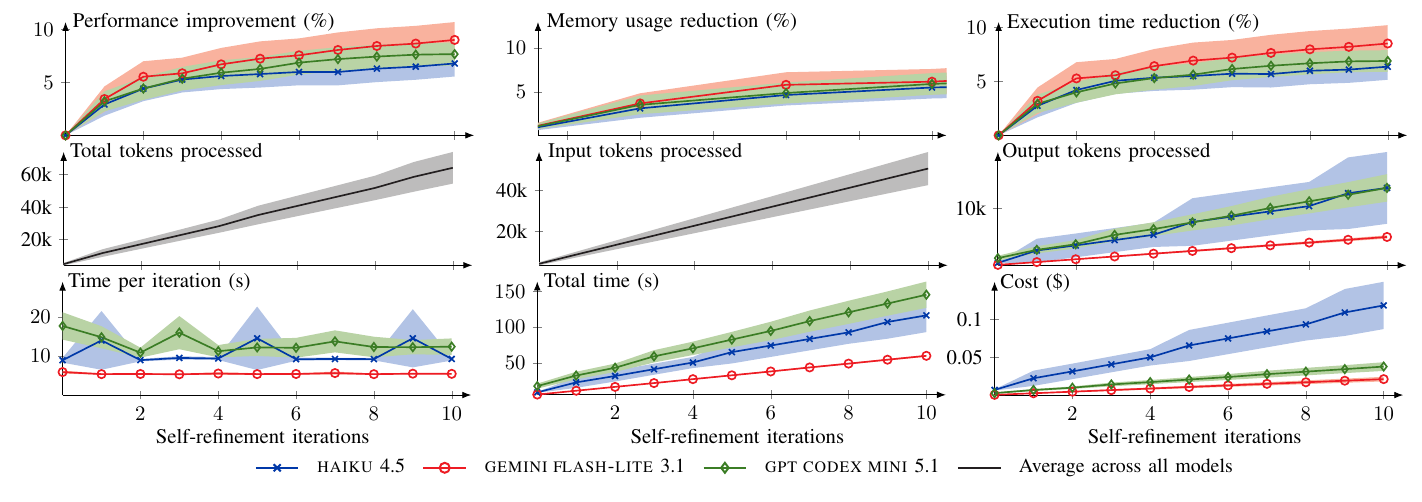}
  }
\caption{Empirical measurements of the self-refinement process across three frontier models on \textsc{effibench} \cite{effibench}. The top row shows performance metrics (execution time and memory usage) as functions of the refinement iteration. The middle row shows cumulative token consumption (total, input, and output). The bottom row shows the refinement time (per iteration and cumulative) and cumulative monetary cost. Curves show the mean values and shaded regions indicate one standard deviation. The cost is computed from per-token charges as listed by each model provider on March 17, 2026.}
  \label{fig:iterations}
\end{figure*}
\section{Identification of the Stopping Problem}\label{sec:id}
The optimal stopping problem formulated in the preceding section is defined by three components: the transition dynamics $f$, the continuation cost $c$, and the payoff function $g$; cf.~\eqref{eq:objective}. In this section, we identify these components based on empirical observations of three foundation models. 

\subsection{Example Use Case: Code Optimization}
The optimal stopping formulation presented in \eqref{eq:objective} is general and applies to any task that a foundation model can be used to solve, provided that the output of the model can be associated with a quality score. However, for the purpose of analysis and experimental validation, we instantiate the stopping problem for a specific \textit{code optimization use case}. 

In this use case, the input to the foundation model is a programming task and the output is code for solving the task. At each refinement iteration, the generated code is executed against a test suite that measures its correctness, execution time, and memory usage. These measurements are aggregated into a candidate score $\tilde{x}_k \in [0,1]$ defined as 
\begin{subequations}\label{eq:score}
\begin{align}
  \tilde{x}_k &=
        \begin{dcases}
          0, & \text{if tests fail},\\
          \operatorname{clip}\!\left(\frac{T_{\max} - T_k}{T_{\max} - T_{\min}},\; 0,\; 1 \right), & \text{otherwise},
        \end{dcases}\\
  T_k &= \frac{1}{T_{\mathrm{ref}}} \int_0^{\mathcal{T}_k} m_k(t)\, \mathrm{d}t,
\end{align}
\end{subequations}
where $m_k(t)$ is the memory footprint at time~$t$ when executing the code generated at iteration $k$, $\mathcal{T}_k$ is the total execution time, $\operatorname{clip}(\cdot,0,1)$ restricts the input to the unit interval, and $T_{\mathrm{ref}}$ is the memory-over-time integral of a given reference code, as defined in \cite{effibench}. The ratio~$T_k$ thus captures both execution time and memory consumption in a single metric. The constants $T_{\min}$ and $T_{\max}$ are normalization bounds chosen so that code that passes the tests with $T_k \leq T_{\min}$ receives the maximal score $\tilde{x}_k = 1$, while failing code and code that passes the tests with $T_k \geq T_{\max}$ receive $\tilde{x}_k = 0$.

After computing the candidate score $\tilde{x}_k$ through \eqref{eq:score}, we update the state of the stopping problem, which is the score of the best code retained so far, i.e., $x_k = \max\{x_{k-1}, \tilde{x}_k\}$ with initial state $x_0=\tilde{x}_0$. Next, we feed the updated state $x_k$ to the stopping policy $\mu$, which decides whether to stop or continue the refinement process by generating new code.
\subsection{Data Collection}
We collect data using \textsc{effibench} \cite{effibench}, which is a coding benchmark for large language models. This benchmark consists of programming tasks for which solutions are assessed based on the candidate score defined in \eqref{eq:score}. We apply three frontier foundation models to this benchmark: \textsc{haiku} 4.5 \cite{anthropic2024claude3}, \textsc{gemini flash-lite} 3.1 \cite{geminiteam2024geminifamilyhighlycapable}, and \textsc{gpt codex mini} 5.1 \cite{openai2024gpt4technicalreport}. For each model, we run the self-refinement loop illustrated in Fig.~\ref{fig:method} for $N=10$ iterations across $50$ tasks.

Figure~\ref{fig:iterations} (on the next page) summarizes the collected data. The top row shows that all three models improve code efficiency through self-refinement, but with diminishing returns. Most of the improvement occurs within the first $3-4$ iterations. Moreover, the middle and bottom rows in Fig.~\ref{fig:iterations} show that token consumption and monetary cost grow approximately linearly with the number of refinement iterations. Together, these trends illustrate the trade-off between the quality of the generated output and the cost of refinement.

\subsection{System Identification}\label{sec:sys_id}
We use the collected data to estimate the parameters of the stopping problem as follows. First, we define the continuation cost $c$ [cf.~\eqref{eq:objective}] to be the average monetary cost of performing a self-refinement iteration. This cost depends on the choice of foundation model and is defined in Table~\ref{tab:cost}.
\begin{table}[H]
  \centering
  \scalebox{0.75}{
    \begin{tabular}{ccc} \toprule
\rowcolor{lightgray}
      {\textit{Foundation model}} & {\textit{Continuation cost $c$}} & {\textit{Explanation}} \\ \midrule
      \textsc{haiku} 4.5 & 0.01 & Avg. cost (\$) per iteration; see Fig.~\ref{fig:iterations}.\\
      \textsc{gemini flash-lite} 3.1 & 0.0025 & Avg. cost (\$) per iteration; see Fig.~\ref{fig:iterations}.\\
      \textsc{gpt codex mini} 5.1 & 0.005 & Avg. cost (\$) per iteration; see Fig.~\ref{fig:iterations}.\\            
  \end{tabular}}
  \caption{The continuation cost for the three foundation models. (Computed based on per-token charges as of March 17, 2026.)}\label{tab:cost}
\end{table}
Second, we define the payoff function $g$ as $g(x)=\beta x$, where $\beta > 0$ is a weighting factor that allows controlling the relative importance of output quality compared to refinement cost. This parameter is not estimated from data; rather, it is a configuration parameter that can be adapted to the use case.

Lastly, we estimate the system function $f$ [cf.~\eqref{eq:dynamics}] from empirical observations of the refinement process. In particular, we model the transition dynamics as a regression problem where the candidate score $\tilde{x}_{k+1}$ [cf.~\eqref{eq:score}] is predicted from the current state $x_k$; the state transition then follows from $x_{k+1}=\max\{x_k,\tilde{x}_{k+1}\}$. To capture both the nonlinear relationship and the inherent uncertainty in the refinement process, we estimate $f$ using a Gaussian process (GP) model.

Formally, we assume that $f$ can be represented as
\begin{align}
x_{k+1} &= f(x_k, w_k) = \min\{1, \max\{x_k, q(x_k) + w_k\}\},\label{eq:q_def}
\end{align}
where $k = 0,1,\ldots, N-1$ is the stage of the refinement process, $q$ is an unknown \textit{score function}, and $w_k$ is a zero-mean Gaussian noise variable with variance $\sigma^2$. The $\max$ operator in \eqref{eq:q_def} reflects the fact that the previous code is retained if the newly generated code achieves a lower score. Similarly, the $\min$ operator is used to ensure that $x \in [0,1]$.

To estimate the unknown function $q$, we use a Gaussian process prior, which allows us to obtain a posterior distribution over functions conditioned on the observed data. Specifically, we define the prior as $q \sim \mathcal{GP}(m, \kappa)$, where $m$ is the mean function and $\kappa$ is the covariance function of the GP. In this paper, we define these functions as
\begin{subequations}\label{gp_prior}
\begin{align}
m(x) &=x, && \\
\kappa(x,x') &=\left(1 + \sqrt[]{5}r + \frac{5r^2}{3}\right)\exp\left(-\sqrt[]{5}r\right),\label{eq:covariance}
\end{align}
\end{subequations}
for all states $x,x' \in [0,1]$, where $r = |x-x'|$. This covariance function encodes the assumption that $q$ varies smoothly over the input space. Similarly, the mean function encodes that refinement preserves the current quality score.

Given this prior, we update it using the empirical observations from evaluations on \textsc{effibench} \cite{effibench} (see Fig.~\ref{fig:iterations}) and Bayes' rule. (We set $\sigma^2$ equal to the variance estimated during the GP fit.)\footnote{See Rasmussen and Williams for detailed formulas of the Bayesian updates \cite[Def.~2.1]{Rasmussen2006Gaussian}. Due to clipping [cf.~\eqref{eq:score}], a candidate score $\tilde{x}_{k+1} \in \{0,1\}$ may correspond to a value of $q(x_k)+w_k$ outside $[0,1]$. We nevertheless treat all candidate scores as direct observations of $q(x_k)+w_k$, which yields a Gaussian approximation of the posterior.} We denote the mean of the resulting posterior as $\tilde{q}(x)=E\{q(x) \mid \text{observed data}\}$. Figure~\ref{fig:gps} shows the posteriors for the three foundation models. We observe that the model $\tilde{q}(x)$ exhibits diminishing returns as the state $x$ increases. In particular, for lower-quality scores (states), the gain from refinement is relatively large, while for higher-quality scores, the improvement becomes smaller. 

\begin{figure}[H]
  \centering
  \scalebox{0.75}{
    \includegraphics{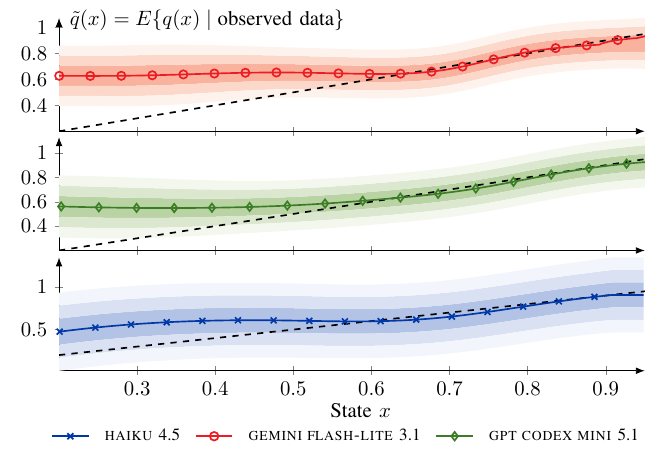}
  }
\caption{Estimation of the score function $q$ [cf.~\eqref{eq:q_def}] when applying the foundation models to \textsc{effibench} \cite{effibench}. Curves show the mean values of the Gaussian processes; shaded regions indicate one, two, and three standard deviations (darker to lighter shades). The dashed lines show $\tilde{q}(x)=x$.}
  \label{fig:gps}
\end{figure}

From the identified models in Fig.~\ref{fig:gps}, we extract the following two structural observations of the system dynamics. 

\begin{observation}[Monotonicity]\label{obs:monotonicity}
The identified model $\tilde{q}$ [cf.~\eqref{eq:q_def}] is nondecreasing on $[0,1]$; see Fig.~\ref{fig:gps}. Since the pointwise minimum and maximum of two nondecreasing functions are also nondecreasing, it follows that $f(x,w) = \min\{1,\max\{x,\, \tilde{q}(x) + w\}\}$ is nondecreasing in $x$.
\end{observation}
\begin{observation}[Diminishing returns]\label{obs:diminishing}
The difference $\tilde{q}(x)-x$ [cf.~\eqref{eq:q_def}] is nonincreasing on $[0,1]$; see Fig.~\ref{fig:gps}. Since the operations $\min\{1-x,\hdots\}$ and $\max\{0,\hdots\}$ preserve this nonincreasing structure, it follows that $f(x,w)-x=\min\{1-x, \max\{0,\tilde{q}(x)-x+w\}\}$ is nonincreasing in $x$.
\end{observation}
These observations will be exploited in the next section to derive theoretical properties of an optimal stopping policy.
\section{Structure of an Optimal Stopping Policy}\label{sec:structural_properties}
Leveraging the identified system dynamics (see Fig.~\ref{fig:gps}), we now derive structural properties of an optimal stopping policy. In the following analysis, we treat Observations~\ref{obs:monotonicity}--\ref{obs:diminishing} as exact characteristics of the system function $f$ [cf.~\eqref{eq:dynamics}] and assume the linear payoff structure defined in \S\ref{sec:sys_id}. That is, we model the GP posterior mean $\tilde{q}$ (which is continuous) as the score function $q$ in \eqref{eq:q_def} and define the payoff function $g$ as $g(x)=\beta x$, where $\beta > 0$ is a weighting factor.

\begin{lemma}\label{lem:value_props}
For each stage $k \in \{0,1,\hdots,N\}$,
\begin{enumerate}
\item The value function $V_k^{\star}(x)$ is nondecreasing in $x$.
\item The difference $V_k^{\star}(x) - g(x)$ is nonincreasing in $x$.
\item The functions $V_k^{\star}$ and $Q^{\star}_k$ (for $k<N$) are continuous.  
\end{enumerate}  
\end{lemma}
\begin{proof}
We proceed by induction on $k$.

\emph{Base case}: $V_N^{\star}(x)=g(x)=\beta x$ is nondecreasing and continuous, and $V_N^{\star}(x)-g(x)=0$ is nonincreasing.

\emph{Inductive step}: Assume that $V_{k+1}^{\star}$ is nondecreasing and continuous, and that $V_{k+1}^{\star}(y)-g(y)$ is nonincreasing in $y$. Since $f(x,w)$ is nondecreasing in $x$ by Observation~\ref{obs:monotonicity} and $V_{k+1}^{\star}$ is nondecreasing by the inductive hypothesis, $Q^{\star}_k(x) = -c + E\{V_{k+1}^{\star}(f(x, w_k))\}$ is nondecreasing in $x$. Since $g(x)$ is also nondecreasing, $V_k^{\star}(x)=\max\{g(x),\, Q^{\star}_k(x)\}$ is nondecreasing. This completes the proof of the first property.

For the second property, we start by expanding $V_k^{\star}(x)-g(x)$ as $V_k^{\star}(x)-g(x)=\max\{0,\, Q^{\star}_k(x)-g(x)\}$, which implies that it suffices to show that $Q^{\star}_k(x)-g(x)$ is nonincreasing. We expand this expression as
\begin{align*}
Q^{\star}_k(x) - g(x) &= -c + E\big\{V_{k+1}^{\star}(f(x, w_k)) - g(f(x,w_k))\big\}\\
&\quad  + \beta\, E\big\{f(x, w_k) - x\big\}.
\end{align*}
By the inductive hypothesis, $V_{k+1}^{\star}(y)-g(y)$ is nonincreasing in $y$. Since $f(x,w)$ is nondecreasing in $x$ by Observation~\ref{obs:monotonicity}, the first expectation is nonincreasing in $x$. By Observation~\ref{obs:diminishing}, $f(x, w_k)-x$ is nonincreasing in $x$, and so is the second expectation. As the sum of nonincreasing functions is nonincreasing, $Q^{\star}_k(x)-g(x)$ is nonincreasing in $x$.

For the third property, $f(x,w)=\min\{1,\max\{x,\,q(x)+w\}\}$ is continuous in $x$ for each $w$ since $q$ is continuous. Hence, since $V_{k+1}^{\star}$ is continuous by the inductive hypothesis, $x \mapsto V_{k+1}^{\star}(f(x,w))$ is continuous and bounded for each $w$. As a consequence, $Q^{\star}_k(x)=-c+E\{V_{k+1}^{\star}(f(x,w_k))\}$ is continuous by the dominated convergence theorem. Therefore, $V_k^{\star}=\max\{g,\,Q^{\star}_k\}$ is the maximum of two continuous functions, which implies that $V_k^{\star}$ is continuous.
\end{proof}

Given the above lemma, there exists an optimal policy with threshold structure, as formally stated below.

\begin{proposition}\label{prop:threshold}
There exist an optimal policy $\mu^{\star}$ and thresholds $\alpha_0,\alpha_1, \hdots, \alpha_{N-1} \in [0,1]$ such that
\begin{equation}\label{eq:threshold_policy}
  \mu_k^{\star}(x) =
  \begin{dcases}
    \mathsf{S} & \text{if }x \geq \alpha_k,\\
    \mathsf{C} & \text{if }x < \alpha_k,
  \end{dcases}
\quad\quad\quad k=0,1,\hdots,N-1.  
\end{equation}
\end{proposition}
\begin{proof}
The optimal policy stops at stage $k$ if and only if $g(x) \geq Q^{\star}_k(x)$. By Lemma~\ref{lem:value_props}, $Q^{\star}_k(x)-g(x)$ is nonincreasing in $x$. Therefore, if $g(x') \geq Q^{\star}_k(x')$ for some $x' \in [0,1]$, then $g(x) \geq Q^{\star}_k(x)$ for all $x \geq x'$. This means that the stopping set $\mathscr{S}_k^{\star} = \{x \mid x \in [0,1], g(x) \geq Q^{\star}_k(x)\}$ is an upper set, i.e., if $x \in \mathscr{S}_k^{\star}$ then $x'\in \mathscr{S}_k^{\star}$ for all $x' \in [x,1]$. It remains to show that $\mathscr{S}_k^{\star}$ is nonempty and closed. For nonemptiness, note that $f(1,w)=1$ for all $w$ [cf.~\eqref{eq:q_def}]. Hence, any policy that continues $j \geq 0$ times from state $1$ receives $g(1)-jc$, which is maximized at $j=0$. Therefore $V^{\star}_{k+1}(1)=g(1)$, which gives $Q^{\star}_k(1)=g(1)-c<g(1)$ and thus $1 \in \mathscr{S}_k^{\star}$. Further, $\mathscr{S}_k^{\star}$ is closed since $g-Q^{\star}_k$ is continuous by Lemma~\ref{lem:value_props}. A nonempty closed upper subset of $[0,1]$ has the form $[\alpha_k,1]$ with $\alpha_k = \min \mathscr{S}_k^{\star}$, which yields \eqref{eq:threshold_policy}.
\end{proof}

\begin{proposition}\label{prop:monotone_thresholds}
The optimal thresholds [cf.~\eqref{eq:threshold_policy}] are nonincreasing in the stage $k$, i.e., $\alpha_0 \geq \alpha_1 \geq \cdots \geq \alpha_{N-1}$.
\end{proposition}
\begin{proof}
We start by showing that $V_k^{\star}(x) \geq V_{k+1}^{\star}(x)$ for all stages $k \in \{0,1,\hdots,N-1\}$ and states $x \in [0,1]$.

\emph{Base case} ($k=N-1$): We have $V_{N-1}^{\star}(x) = \max\{g(x),\, -c + E\{g(f(x,w))\}\} \geq g(x) = V_N^{\star}(x)$.

\emph{Inductive step}: Assume $V_{k+1}^{\star}(x) \geq V_{k+2}^{\star}(x)$ for all states $x \in [0,1]$. Then
\begin{align*}
  Q^{\star}_k(x) &= -c + E\{V_{k+1}^{\star}(f(x,w))\} \\
                 &\geq -c + E\{V_{k+2}^{\star}(f(x,w))\} \\
         &= Q^{\star}_{k+1}(x).
\end{align*}
Consequently,
\begin{align*}
  V_k^{\star}(x) &= \max\{g(x),\, Q^{\star}_k(x)\} \\
                 &\geq \max\{g(x),\, Q^{\star}_{k+1}(x)\} \\
                 &= V_{k+1}^{\star}(x).
\end{align*}  
Since $Q^{\star}_k(x) \geq Q^{\star}_{k+1}(x)$, we have
\begin{align*}
  \mathscr{S}_k^{\star} &= \{x \mid x \in [0,1], g(x) \geq Q^{\star}_k(x)\} \\
                        &\subseteq \{x \in [0,1] \mid g(x) \geq Q^{\star}_{k+1}(x)\} = \mathscr{S}_{k+1}^{\star}.
\end{align*}
Taking minima of $\mathscr{S}_k^{\star}$ yields $\alpha_k \geq \alpha_{k+1}$.
\end{proof}
The preceding propositions imply that each stopping set $\mathscr{S}_k^{\star}$ has the form $[\alpha_k, 1]$ for some threshold $\alpha_k \in [0,1]$. We now show that these sets are all equal.
\begin{figure*}
  \centering
  \scalebox{0.75}{
    \includegraphics{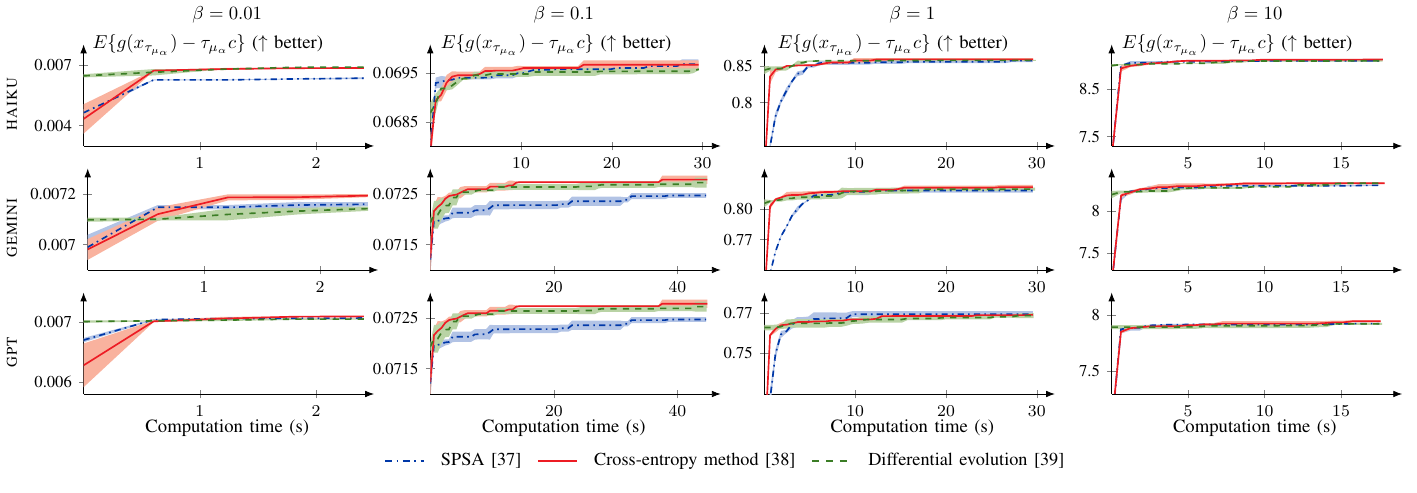}
  }\caption{Convergence curves of three simulation-based optimization methods (SPSA \cite{spsa}, Cross-entropy method \cite{rubinstein1999cross}, and Differential evolution \cite{differential_evolution}) for computing the optimal threshold $\alpha^{\star}$ [cf.~Prop.~\ref{prop:closed}] across the three foundation models and four values of the payoff weight $\beta$ (defined in \S\ref{sec:sys_id}). Curves show mean values and shaded regions indicate one standard deviation over 3 runs with different random seeds. The y-axes indicate the expected value $E\{g(x_{\tau_{\mu_{\alpha}}})-c\tau_{\mu_{\alpha}}\}$ (i.e., the expected payoff minus the refinement costs) and the x-axes indicate the computation times in seconds.}\label{fig:numerical}
\end{figure*}
\begin{proposition}\label{prop:closed}
The optimal stopping sets satisfy $\mathscr{S}_0^{\star} = \mathscr{S}_1^{\star} = \cdots = \mathscr{S}_{N-1}^{\star}$. Equivalently, the optimal thresholds are stage-independent, i.e., $\alpha_0 = \alpha_1 = \cdots = \alpha_{N-1} = \alpha^{\star}$.
\end{proposition}
\begin{proof}
By Prop.~\ref{prop:monotone_thresholds}, we have $\mathscr{S}_k^{\star} \subseteq \mathscr{S}_{k+1}^{\star}$ for all stages $k \in \{0,1,\hdots,N-2\}$. It remains to show the reverse inclusion. Fix $k \in \{0,1,\hdots,N-2\}$ and let $x \in \mathscr{S}_{k+1}^{\star}$, so that $x \geq \alpha_{k+1}$ and $g(x) \geq Q^{\star}_{k+1}(x)$. Since $f(x,w) = \min\{1, \max\{x,\, q(x)+w\}\} \geq x \geq \alpha_{k+1} \geq \alpha_{k+2}$ (where $\alpha_N=0$ since stopping is enforced at stage $N$), we have $f(x,w) \in \mathscr{S}_{k+2}^{\star}$ for all realizations of $w$, and therefore $V_{k+2}^{\star}(f(x,w)) = g(f(x,w))$. Similarly, $f(x,w) \in \mathscr{S}_{k+1}^{\star}$, so $V_{k+1}^{\star}(f(x,w)) = g(f(x,w))$. It follows that
\begin{align*}
Q^{\star}_k(x) &= -c + E\{V_{k+1}^{\star}(f(x,w))\} = -c + E\{g(f(x,w))\} \\
        &= -c + E\{V_{k+2}^{\star}(f(x,w))\} = Q^{\star}_{k+1}(x).
\end{align*}
Since $g(x) \geq Q^{\star}_{k+1}(x) = Q^{\star}_k(x)$, we have $x \in \mathscr{S}_k^{\star}$. Thus $\mathscr{S}_{k+1}^{\star} \subseteq \mathscr{S}_k^{\star}$, which together with $\mathscr{S}_k^{\star} \subseteq \mathscr{S}_{k+1}^{\star}$ gives $\mathscr{S}_k^{\star} = \mathscr{S}_{k+1}^{\star}$ for every $k \in \{0,1,\hdots,N-2\}$.
\end{proof}
Propositions~\ref{prop:threshold}--\ref{prop:closed} establish that, given the identified dynamics in Fig.~\ref{fig:gps} and assuming that Observations~\ref{obs:monotonicity}--\ref{obs:diminishing} are exact properties of the system function, an optimal stopping policy is characterized by a single threshold $\alpha^{\star} \in [0,1]$. This structure implies a simple method for computing an optimal policy: parameterize the stopping policy $\mu$ by a threshold $\alpha$ [cf.~\eqref{eq:threshold_policy}] and estimate the optimal threshold $\alpha^{\star}$ by maximizing the expected value $E\{g(x_{\tau_{\mu_{\alpha}}}) - c\tau_{\mu_{\alpha}}\}$, where $\mu_{\alpha}$ denotes the parameterized policy. This optimization can be efficiently performed using simulation-based stochastic approximation methods, as detailed in the following section.
\section{Computing an Optimal Stopping Policy}
We now exploit the structural properties established in the preceding section to compute an optimal stopping policy. By Props.~\ref{prop:threshold}--\ref{prop:closed}, this computation reduces to optimizing the stopping threshold $\alpha$, which we use to parameterize a threshold-based stopping policy $\mu_{\alpha}$; cf.~\eqref{eq:threshold_policy}. We implement this optimization using three methods: simultaneous perturbation stochastic approximation (SPSA) \cite{spsa}, the cross-entropy method \cite{rubinstein1999cross}, and differential evolution \cite{differential_evolution}. For each method, the value $E\left\{g(x_{\tau_{\mu_{\alpha}}}) -c\tau_{\mu_{\alpha}} \right\}$ is estimated via simulation of the Gaussian process model identified in \S\ref{sec:sys_id}. We run all simulations and optimizations on an M4 Pro Chip with macOS Sequoia 15.6.1 and Python 3.11.

Figure~\ref{fig:numerical} shows the convergence curves for varying values of the payoff weight $\beta$, as defined in \S\ref{sec:sys_id}. All methods converge to similar expected values, which suggests that they identify optimal or near-optimal thresholds. In all cases, convergence is achieved within seconds, which demonstrates the computational efficiency afforded by the threshold structure.

Figure~\ref{fig:threshold_bars} shows the optimized thresholds for each model and payoff weight $\beta$. We observe that the thresholds increase with $\beta$, which is expected. As the relative importance of payoff grows, the policy becomes more selective and demands higher scores before stopping. For small $\beta$ (e.g., $\beta=0.01$), the thresholds are low, indicating that the cost of refinement dominates and early stopping is preferred. Across the three foundation models, \textsc{haiku}~4.5 admits lower thresholds than the other models for small $\beta$, which is consistent with its higher continuation cost; cf.~Table~\ref{tab:cost}. For large $\beta$, however, \textsc{haiku}~4.5 exhibits higher thresholds than the other models.

\begin{figure}[H]
  \centering
  \scalebox{0.7}{    
\includegraphics{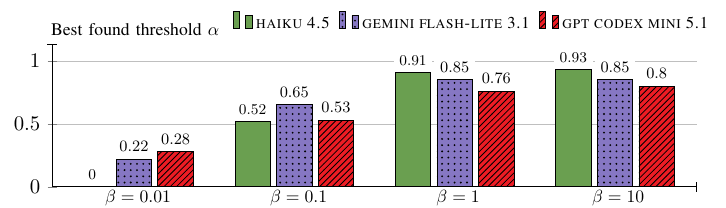}
  }
\caption{The optimized threshold $\alpha$ across the three foundation models and four values of the payoff weight $\beta$ (defined in \S\ref{sec:sys_id}).}\label{fig:threshold_bars}
\end{figure}

\begin{figure*}
  \centering
  \scalebox{0.705}{
    \includegraphics{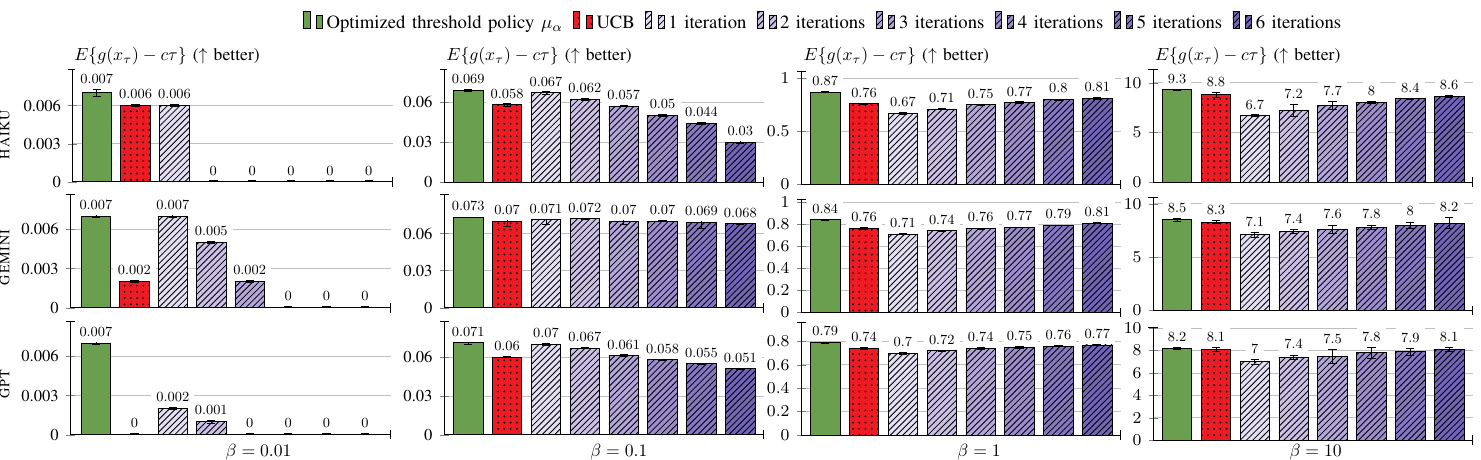}
  }
\caption{Evaluation results on \textsc{effibench} \cite{effibench}. Each group of bars shows the expected value $E\{g(x_{\tau}) - c\tau\}$ for the optimized threshold policy $\mu_{\alpha}$, the UCB baseline, and fixed-iteration stopping policies with $1$ to $6$ iterations. Rows correspond to the three foundation models. Columns indicate the payoff weight $\beta$ (defined in \S\ref{sec:sys_id}). Bar heights indicate the mean over evaluations with three different random seeds; error bars indicate one standard deviation.}\label{fig:effibench_eval}
\end{figure*}
\section{Experimental Evaluation on Effibench}
In this section, we evaluate the optimized stopping policy described in the preceding section on \textsc{effibench} \cite{effibench} and compare it against stopping policies proposed in prior work.

\vspace{2mm}
\noindent\textit{\textbf{Baseline stopping policies.}} We consider two categories of baselines. The first category consists of \textit{fixed-iteration} policies, which stop after a predetermined number of iterations. This is the most common approach in the self-refinement literature; see e.g., \cite{NEURIPS2024_99c66755,madaan2023selfrefine,zelikman2024selftaught,10610065,pich_llms}. The second category is the \textit{upper confidence bound} (UCB) policy proposed by Sun et al.~\cite{sun2026stop}, which adaptively decides when to stop based on confidence bounds on the values of a finite set of stopping thresholds. We instantiate this policy by defining the set of thresholds as $\alpha \in \{0, 0.1, 0.2, \hdots, 1\}$ and selecting the threshold with the highest upper confidence bound. 

\vspace{2mm}
\noindent\textit{\textbf{Experiment setup.}} We evaluate all stopping policies on $50$ problems from \textsc{effibench} \cite{effibench}. (These problems are disjoint from the problems used for system identification in \S\ref{sec:sys_id}.) Each stopping policy is applied to the three foundation models with a maximum of $N=10$ refinement iterations. We run each method three times with different random seeds and report the average performance.

\vspace{2mm}\noindent\textit{\textbf{Evaluation results.}} Figure~\ref{fig:effibench_eval} shows the expected value $E\{g(x_{\tau_{\mu}}) - c\tau_{\mu}\}$ for each stopping policy $\mu$ across the three foundation models and four values of the payoff weight $\beta$. We observe that the optimized threshold policy $\mu_{\alpha}$ achieves the highest expected value in all configurations. 

For the fixed-iteration policies, we find that the best policy varies with $\beta$. In particular, we observe that fewer iterations perform better when $\beta$ is small, while more iterations are preferred when $\beta$ is large. Finally, our experimental results show that the UCB policy automatically adapts to $\beta$ but consistently underperforms the optimized threshold policy.

\section{Discussion of the Experimental Results}
The main takeaways from the theoretical and experimental results can be summarized as follows.
\begin{enumerate}
\item \textit{Self-refinement exhibits diminishing returns.} The empirical measurements in Fig.~\ref{fig:iterations} show that all three foundation models improve output quality through self-refinement, but with progressively smaller gains per refinement iteration. Meanwhile, token consumption and monetary cost grow approximately linearly.

\item \textit{Optimal threshold-based stopping policies exist.} Propositions~\ref{prop:threshold}--\ref{prop:closed} establish that, under the conditions in Observations~\ref{obs:monotonicity}--\ref{obs:diminishing}, an optimal stopping policy is characterized by a single, stage-independent threshold $\alpha^{\star}$. This reduces the policy search from a sequence of functions $\mu_0, \hdots, \mu_{N-1}$ to a scalar optimization problem, which can be efficiently solved; see Fig.~\ref{fig:numerical}.

\item \textit{Optimal stopping improves over the state-of-the-art.} The most common approach in the self-refinement literature is to stop after a fixed number of iterations. As shown in Fig.~\ref{fig:effibench_eval}, the optimized threshold policy consistently outperforms this approach.
\end{enumerate}
\section{Conclusion}
Foundation models can improve their outputs through a self-refinement process where they iteratively critique and refine their own outputs. We show that this process can be formulated as an optimal stopping problem in which the decision to continue refining or to stop is made sequentially based on the expected improvement relative to the refinement cost. We establish conditions for optimal threshold-based stopping policies and validate them on a coding benchmark across three frontier models. Empirical results demonstrate that our stopping policy significantly improves cost-efficiency compared to the state-of-the-art. These results show the value of leveraging decision theory for efficient control of foundation models. As these models are increasingly deployed in large-scale automation pipelines with thousands of daily invocations, we believe that our optimal stopping approach can translate into substantial monetary savings.

\vspace{2mm}
\noindent\textit{\textbf{Future work.}} A natural direction for future work is to evaluate our approach on benchmark types other than coding benchmarks (e.g., mathematical reasoning benchmarks) and to instantiate it with foundation models other than large language models (e.g., time-series models or multi-modal models). From a theoretical perspective, our structural results (Props.~\ref{prop:threshold}--\ref{prop:closed}) rely on Observations~\ref{obs:monotonicity}--\ref{obs:diminishing} holding exactly. Since the system function is estimated via a Gaussian process, a promising direction for further analysis is to relax these assumptions by exploiting the posterior uncertainty of the GP to derive probabilistic guarantees on the threshold structure.

\bibliographystyle{IEEEtran}
\bibliography{references}
\end{document}